\documentclass[12pt,reqno]{amsart}
\usepackage{stmaryrd}
\usepackage{amsmath}
\usepackage{amssymb}
\usepackage{txfonts}
\usepackage{dsfont, amssymb,amsmath,amscd,latexsym, amsthm, amsxtra,amsfonts}
\usepackage[all]{xy}
\usepackage{graphicx}
\usepackage{float}
\newtheorem{theorem}{Theorem}[section]

\newtheorem{example}[theorem]{Example}

\newtheorem{Them}{Theorem}[section]
\newtheorem{Prop}{Proposition}[section]
\newtheorem{Remark}{Remark}[section]
\newtheorem{Lemma}{Lemma}[section]

\newtheorem{CR}{Corollary}[section]

\begin{document}
\baselineskip 17pt
\makeatletter
\def\@setauthors{%
\begingroup
\def\thanks{\protect\thanks@warning}%
\trivlist \centering\footnotesize \@topsep30\p@\relax
\advance\@topsep by -\baselineskip
\item\relax
\author@andify\authors
\def\\{\protect\linebreak}%
{\authors}%
\ifx\@empty\contribs \else ,\penalty-3 \space \@setcontribs
\@closetoccontribs \fi
\endtrivlist
\endgroup } \makeatother
\title[{{ Stock loan model with Automatic termination clause }}]
 {{Stock loan with Automatic termination clause, cap and margin }}
\author[{ Shuqing Jiang,  Zongxia Liang and  Weiming Wu } ]
{Shuqing Jiang \\ Department of Mathematical Sciences, Tsinghua
University, Beijing 100084, China. Email:  jiangsq06@gmail.com\\
Zongxia Liang \\ Department of Mathematical Sciences, Tsinghua
University, Beijing 100084, China. Email:
zliang@math.tsinghua.edu.cn\\
Weiming Wu\\ Department of Mathematical Sciences, Tsinghua
University, Beijing 100084, China. Email:  wuweiming03@gmail.com
 }
\maketitle
\begin{abstract}
This paper works out  fair values of  stock loan model with
automatic termination clause, cap and margin. This stock loan is
treated as a generalized perpetual American option with possibly
negative interest rate and some constraints. Since it helps a bank
to control the risk, the banks charge less service fees compared to
stock loans without any constraints. The automatic termination
clause, cap and margin are in fact a stop order set by the bank.
Mathematically, it is a kind of optimal stopping problems arising
from the pricing of financial products which is first revealed. We
aim at establishing explicitly the value of such a loan and ranges
of fair values of key parameters : this loan size, interest rate,
cap, margin and fee for providing such a service and quantity of
this automatic termination clause and relationships among these
parameters as well as the optimal exercise times. We present
numerical results and make analysis about the model parameters and
how they impact on value of stock loan. \vskip 10 pt
 \noindent { MSC}(2000): primary 91B24, 91B28,91B70 secondary 60H05, 60H10
\end{abstract}
\vskip 15pt
 \noindent
\setcounter{equation}{0}
\section{{\small {\bf Introduction}}}
\vskip 5pt \noindent
 A stock loan is a popular financial product
provided by many banks and financial institutions in which a client
(borrower), who owns one share of stock, borrows a loan of amount
$q$ from a bank (lender) with the share of stock as collateral, and
the bank receives amount $c$ from the client as the service fee. The
client may regain the stock by repaying the principal and interest
(that is, $qe^{\gamma t}$, where $\gamma$ is continuously
compounding loan interest rate ) to the bank at any time $t$, or
surrender the stock instead of repaying the loan.  The key point of
making the stock loan contract is to find  values of the parameters
$q$, $c$, and $\gamma $. The stock loan has many advantages for the
client. It creates liquidity while overcoming the barrier of large
block sales, such as triggering tax events or controlling
restrictions on sales of stocks. It also serves as a hedge against a
market down turn : if the stock price goes down, the client may just
forfeit the stock and does not repay the loan; if however the stock
price goes up, the client keeps all the benefits upside by repaying
the principal and interest. In other words, a stock loan can help
high-net-worth investors with large equity positions to achieve a
variety of objectives.\vskip 15pt\noindent The stock loan valuation
is essentially a kind of  optimal stopping problems. A typical and
well-known example of optimal stopping problems is the American
option. There are many literatures about the American option, we
refer the readers to Hull \cite{Hull}, Gerber and Shiu \cite{Ger}
and Broadie and Detemple \cite{Bro}, Jiang \cite{Jia}, Detemple et
al. \cite{Det},  Cheuk and Vorst \cite{CHeuk}, Windcliff et al.
\cite{Win}, and Dai et al. \cite{Dai}. Stock loan valuation has
attracted much interest of both academic researchers and financial
institution recently. Xia and Zhou \cite{stock} first studied the
problem of stock loan under the Black-Scholes framework. They
established stock loan model and got its valuation by a pure
probabilistic approach. They also pointed out that variational
inequality approach can not be directly applied to these kinds of
stock loans.  Zhang and Zhou \cite{Zha} used the variational
inequality approach to solve the stock loan pricing problem treated
in\cite{stock}, and they carried the approach over to the models in
which the underlying stock price follows a geometric Brownian motion
with regime switching(cf.\cite{Zha}). Dai and Xu
\cite{Dai1}considered the valuation of stock loan that the
accumulative dividends may be gained by the borrower or the lender
according  to the provision of the loan.
\vskip 15pt\noindent
 In order to control effectively the risk and make the stock loan
 contract worthwhile
 so that it can provide the writer with protection, the bank and
 client embed an
automatic termination clause, cap $L$ and margin $k$ into the stock
loan. The stock loan can then be terminated via the clause when the
share price is too low, that is, the automatic termination clause is
triggered if and only if the discounted stock price is less than $a$
(i.e., $e^{-\gamma t}S_{t}\leq a$).  Since it helps a bank to
control the risk, the bank should charge less service fee initially
compared to the stock loan without the automatic termination clause.
The bank will terminate a stock loan contract by acquiring the
ownership of the collateral equity and the client will not need to
pay the principle and interest when the automatic termination clause
is triggered at time $t$. Hence,  the client can choose to regain
the stock by repaying the loan principal and interest. The automatic
termination clause can be described by a quantity $a$ $(0<a\leq q)$,
which is also a key point of negotiation between the bank and the
client. Because there is a distinction between what is actuarial
fair value and values as the solution of a mathematical problem, we
need to determine  the fair value of this loan, ranges of fair
values of the parameters $(q, \gamma, c,a, L , k)$ and relationships
among these parameters in some reasonable sense so that the client
and the bank know whether this actuarial value is reasonable( that
is, this value belongs to the ranges and satisfies the
relationships). Therefore, working out this value in this contract
will be a main task in negotiation between the client and the bank
initially. Thus this is a problem of theoretical value finding as
well as practical implication for option pricing. To the best of our
knowledge, there are a few results on this topic have been reported,
we refer the readers to Dai and Xu \cite{Dai1}, Liu and Xu
\cite{Gy}, Xia and Zhou\cite{stock} and Zhang and Zhou\cite{Zha}.
The main purpose of the present paper is to determine the right
values of these parameters $(q, \gamma, c,a, L, K)$: the principal
$q$, the interest rate $\gamma $, the fee $c$ charged by the bank,
the barrier $a$, the cap $L$ and margin $k$ in the stock loan
contract with automatic termination clause and find relationships
among these parameters by deriving optimal exercise time (stopping
time) and valuation formulas of the stock loan under the assumption
$\delta>0$ and $\gamma-r+\delta\geq0$ or $\delta=0$ and
$\gamma-r>\frac{\sigma^{2}}{2}$( where $\delta$ is the dividend
yield, $r$ is the risk-free rate, and $\sigma$ is the volatility).
We try to develop variational inequality method(cf.
\cite{Methods,Oksendal,Liang}) with probabilistic approach to  deal
with
 this value of such a loan and  ranges of fair values of this stock loan
size, interest rate, cap, margin and fee for providing such a
service and quantity of this automatic termination clause and
relationships among these parameters. The paper establishes a
general setting to broaden the applicability of our method
concerning different stock loans. \vskip 15pt \noindent The paper is
organized as following: In section 2, we formulate a mathematical
model of the stock loan with automatic termination clause. In
section 3, we evaluate the stock loan by variational inequality
method and obtain an optimal exercise time. In section 4, we derive
probabilistic solutions and terminable exercise times of the stock
loan. In section 5,we study  a mathematical model of the stock loan
with automatic termination clause, cap and margin by applying the
way we used in the section 3 and section 4 to determine fair values
of the stock loan in section 6. In section 7 we give some numerical
results of two stock loans. In section 8, we give an over view of
the main findings in this paper. In appendix, we further give
discussions of the parameters. \vskip 15pt
 \noindent
 \setcounter{equation}{0}
\section{{\small {\bf Formulation of stock
loan with automatic termination clause}}} \vskip 5pt
 \noindent We introduce in
this section the standard Black-Scholes model in a continuous-time
financial market consisting of two assets: a risky asset stock $S$
and a locally risk-less money account $B\equiv \{B_{t}, t\geq 0 \}$.
The uncertainty is described by a standard Brownian motion $\mathcal
{W}\equiv\{\mathcal {W}_{t}, t\geq 0\}$ defined on a risk-neutral
probability space $(\Omega, \mathcal {F}, \{ \mathcal
{F}_{t}\}_{t\geq 0}, P)$, where $\{ \mathcal {F}_{t}\}_{t\geq 0}  $
is the $P-$augmentation of the filtration generated by $ \mathcal
{W} $, with $\mathcal {F}_{0}=\sigma\{ \Omega, \emptyset\}$ and $
\mathcal {F}=\sigma\{ \bigcup_{t\geq 0}\mathcal {F}_{t} \}$. The
terms {\sl fair value}, {\sl right value} and {\sl proper value},
$\cdots$  in this paper mean that they are determined under this
risk-neutral probability $P$.  The locally risk-less money account
$B$ evolves according to the following dynamic system,
\begin{eqnarray*}
dB_{t}=rB_{t}dt,\   \  r>0.
\end{eqnarray*}
The market price process $S$ of the stock follows a geometric
Brownian motion,
\begin{eqnarray}\label{E2.1}
S_{t}=S_{0}e^{(r-\delta-\frac{\sigma^{2}}{2})t+\sigma\mathcal
{W}_{t}},
\end{eqnarray}
where $S_{0}$ is the initial stock price, $\delta\geq0$ is the
dividend yield and $\sigma >0$ is the volatility. \vskip 15pt
\noindent We now explain the stock loan (i.e., the contract) with an
automatic termination clause in this paper as follows: \vskip 15pt
\noindent $\bullet$ At the beginning,  a client borrows amount $q
(q>0)$ from a bank with one share of stock as the collateral, and
gives the bank amount $c(0\leq c \leq q)$ as the service fee. As a
result, the client gets amount $q-c$ from the bank. \vskip 15pt
\noindent $ \bullet$   The client has the option to regain the stock
by paying amount $qe^{\gamma t}$( where $\gamma$ is the continuously
compounding loan interest rate) to the bank (lender) at any time
$t$, or just gives the stock to the bank without repaying the loan
before triggering the automatic termination clause. Dividends of the
stock are collected by the bank until the client regains the stock,
the dividends are not credited to the client.
 \vskip 15pt \noindent
$\bullet$  The client has no obligation to regain the stock whether
the automatic termination clause is triggered or not. If the
automatic termination clause is triggered, then the bank acquires
the collateral stock,  the contract is terminated, and the client
loses the option to regain the stock. \vskip 15pt \noindent
$\bullet$ The values of $(q, \gamma, c,a)$: the principal $q$, the
interest rate $\gamma $, the fee $c$ charged by the bank, and the
barrier $a$ are specified before this contract is exercised. \vskip
15pt \noindent Xia and Zhou \cite{stock} established a stock loan
without an automatic termination clause  by probabilistic approach.
They proved that the optimal exercise time is a hitting time:
\begin{eqnarray*}
\tau_{b}=\inf{\{t\geq 0,e^{-\gamma t}S_{t} \geq b\}},
\end{eqnarray*}
then determined the value by maximizing  expected discounted payoff
of this stock loan given by $\tau_{b}$ for some $b\geq q\vee S_{0}$,
where $q$ is the principal of the stock loan and $S_{0}$ is the
initial stock price. \vskip 15pt \noindent The automatic termination
clause is one of our main interest. The main goal of sections 3 and
4 is to determine  fair value $f(S_0)$ ( see (2.2) below) of the
stock loan with an automatic termination clause and ranges of fair
values of the parameters $ (q, c, \gamma,a)$ under the assumption
$\delta>0$ and $\gamma-r+\delta\geq0$ or $\delta=0$ and
$\gamma-r>\frac{\sigma^{2}}{2}$ (see  Proposition
\ref{solution_exist} below). This problem can be treated as a
generalized perpetual American option with a client initially buying
at price $S_{0}-q+c$. \vskip 15pt \noindent
 We consider the automatic
termination clause as follows: if the stock price satisfies
$e^{-\gamma t}S_{t} \leq a, 0<a\leq q$ ($\gamma$ is the loan
interest rate), then this stock loan is terminated. So the
discounted payoff of this American contingent claim at stopping time
$\tau \in \mathcal {T}_{0}$ is
\begin{eqnarray*}
Y(\tau)=e^{-r\tau}(S_{\tau}-qe^{\gamma
\tau})_{+}I_{\{\tau<\tau_{a}\}},
\end{eqnarray*}
where $\tau_{a}=\inf{\{t\geq 0,e^{-\gamma t}S_{t} \leq a\}}$ and
$\mathcal {T}_{0}$ denotes all $\{ \mathcal {F}_{t}\}_{t\geq 0}$
-stopping times. The initial value of this American contingent claim
is the following (cf. \cite{Methods,Towards}),
\begin{eqnarray}\label{reward}
f(x)&=&\sup\limits_{\tau \in \mathcal {T}_{0}}{\bf E}\big [Y(\tau)\big ]\nonumber\\
&=&\sup\limits_{\tau \in \mathcal {T}_{0}}{\bf E}\big
[e^{-r\tau}(S_{\tau}-qe^{\gamma
\tau})_{+}I_{\{\tau<\tau_{a}\}}\big ]\nonumber\\
&=&\sup\limits_{\tau \in \mathcal {T}_{0}}{\bf E}\big [e^{-\tilde
{r}\tau}(\tilde{S}_{\tau}-q)_{+}I_{\{\tau<\tau_{a}\}}\big ],
\end{eqnarray}
where $\tilde{r}=r-\gamma\leq 0$ and $\tilde{S}_{t}=e^{-\gamma
t}S_{t},\tilde{S}_{0}=S_{0}=x$. The value of this American
contingent claim at time $t$ is the following,
\begin{eqnarray}
V_{t}=\sup\limits_{\tau \in \mathcal {T}_{t}}{\bf E}\big[e^{-
{r}(\tau-t)}(S_{\tau}-qe^{\gamma
\tau})_{+}I_{\{\tau<\tau_{a}\}}|\mathcal{F}_{t}\big],
\end{eqnarray}
i.e.,
\begin{eqnarray*}
e^{-rt}V_{t}=\sup\limits_{\tau \in \mathcal {T}_{t}}{\bf E}\big
[e^{-\tilde
{r}\tau}(\tilde{S}_{\tau}-q)_{+}I_{\{\tau<\tau_{a}\}}|\mathcal{F}_{t}\big
],
\end{eqnarray*}
where $\mathcal {T}_{t}$ denotes all $\{ \mathcal {F}_{t}\}_{t\geq
0}$ -stopping times $\tau $ with $\tau \geq t$ a.s.. \vskip 15pt
\noindent In the following sections we first determine  fair value
$f(S_0)$ of the stock loan with an automatic termination clause,
then find ranges of fair values of the parameters $ (q, c,
\gamma,a)$ and relationships among these parameters by $f(S_{0})$
and equality $f(S_{0})=S_{0}-q+c$. \vskip15pt
 \noindent
\setcounter{equation}{0}
\section{{\small {\bf Variational inequality method}}}
\vskip 5pt
 \noindent
In this section we compute the fair  value $f(S_0)$ of the stock
loan with an automatic termination clause  treated as a generalized
perpetual American option with automatic termination clause. Note
that since the payoff process of the option $Y(t)\geq 0$ a.s., and
$Y(t)>0$ with a positive probability if $S_{0}>a$, $Y(t)=0$ a.s. if
$S_{0}\leq a$, to avoid arbitrage we assume that
\begin{eqnarray}
 S_{0}-q+c>0,\ S_{0}>a,
\end{eqnarray}
and
\begin{eqnarray}
 S_{0}-q+c=0,\ S_{0}\leq a.
\end{eqnarray}
 Now we introduce some quantitative
properties on $f$ defined via (\ref{reward}) and solve the optimal
stopping time problem (\ref{reward}) by {\sl variational method and
stopping time techniques.}
\begin{Prop}\label{inequality}
$(x-q)_{+}\leq f(x)\leq x$ for all $x\geq 0$.
\end{Prop}
\begin{proof} \  By taking $\tau=0$ in (\ref{reward}) and noticing
that $\tau <\tau_{a}$, a.s.,  it is easy to see that $(x-q)_{+}\leq
f(x)$. As for the second inequality, we have
\begin{eqnarray*}
f(x)&=&\sup\limits_{\tau \in \mathcal {T}_{0}}{\bf E}\big
[e^{-\tilde
{r}\tau}(\tilde{S}_{\tau}-q)_{+}I_{\{\tau<\tau_{a}\}}\big ]\\
&\leq&\sup\limits_{\tau \in \mathcal {T}_{0}}{\bf E}\big [e^{-\tilde
{r}\tau}\tilde{S}_{\tau}I_{\{\tau<\tau_{a}\}}\big ]\\
&\leq&\sup\limits_{\tau \in \mathcal {T}_{0}}{\bf E}\big [e^{-\tilde
{r}(\tau\wedge\tau_{a})}\tilde{S}_{\tau\wedge\tau_{a}}\big ]\\
&=&\sup\limits_{\tau \in \mathcal {T}_{0}}{\bf E}\big
[xe^{\sigma\mathcal{W}_{\tau\wedge
\tau_{a}}-\frac{\sigma^{2}}{2}\tau\wedge\tau_{a}}\big ]\\
&=&x,
\end{eqnarray*}
where the last equality follows from  the optional sampling theorem
and  the process $\{e^{\sigma\mathcal{W}_t-\frac{\sigma^{2}t}{{2}}},
t\geq 0 \}$  is a strong  martingale.
\end{proof}
\begin{Remark}
 It is easy to see from the definition of $f(x)$  that $f(x)$
is continuous,  convex and nondecreasing with respect to $x$.
\end{Remark}
\vskip 5pt \noindent Because the loan rate $\gamma$ is always
greater than risk-free rate $r$, our problem  reduces to a
generalized perpetual American contingent claim  with possibly
negative interest rate $r-\gamma  $, where the term {\sl negative
interest rate } is just used to state relationship between the model
treated in this paper and an American perpetual call option with a
time-varying striking price, and has no other implications. We have
the following. \vskip10pt
 \noindent
\begin{Them}\label{main}
Assume that $\delta>0$ and $\gamma-r+\delta\geq0$ or $\delta=0 $ and
$\gamma-r>\frac{\sigma^{2}}{2}$. If
 $f(x)$  is continuous, $f(x)\in \mathcal
{C}^{1}[(0,\infty)\setminus \{a\}]\cap \mathcal
{C}^{2}[(0,\infty)\setminus \{a,b\}]$ for some $b\geq 0 $ which we
will discuss later,  and $ f(x)$ satisfies the following variational
problem
\begin{eqnarray}\label{variational}
\left\{
\begin{array}{l l}
\max{\{\frac{1}{2}\sigma^{2}x^{2}f^{''}+
(\tilde{r}-\delta)xf^{'}-\tilde{r}f,(x-q)_{+}-f
\}}=0,x>a,\\
f(x)=0,x\leq a,
\end{array}
\right.
\end{eqnarray}
then   $f(x)$ must be the function defined by (\ref{reward} ) and
 $\tau_{b}=\inf{\{t\geq 0:e^{-\gamma t}S_{t}\geq b\}}$ attains the
supremum in (\ref{reward}), i.e., $\tau_{b}$  is optimal.
\end{Them}
\vskip10pt
 \noindent
\begin{Remark}
The value $a$ is always  determined by negotiation between  the bank
and the client initially, the $b$ is an endogenous parameter to be
determined late in this model.
\end{Remark}
\vskip5pt
 \noindent
\begin{proof} Let  $f(x)$ satisfy problem (\ref{variational}), we want to show that $f$ must be the function
defined by (\ref{reward}). Since $f(x)=0$, $0<x\leq a$, we only need
to prove Theorem \ref{main} in the
region $a<x$. We will prove Theorem \ref{main} in two steps.\\
{\bf Step one.} \ We show that for any stopping time $\tau$
\begin{eqnarray}\label{step1}
f(x)\geq {\bf E}\big [e^{-\tilde{r}\tau}(\tilde
{S}_{\tau}-q)_{+}I_{\{\tau<\tau_{a}\}}\big ].
\end{eqnarray}
 Applying
 It\^{o} formula to  convex function $f$ and the process $\tilde{S}_{t} $
 defined in (\ref{reward}) and using (\ref{variational})we have
\begin{eqnarray} \label{Ito1}
d(e^{-\tilde{r}t}f(\tilde{S}_{t}))&=&e^{-\tilde{r}t}
\tilde{S}_{t}f^{'}(\tilde{S}_{t})\sigma
d\mathcal{W}(t)-e^{-\tilde{r}t}\big
[(\delta\tilde{S}_{t}-\tilde{r}q)
I_{\{\tilde{S}_{t}>b\}}\big ]dt\nonumber\\
&  \equiv &d\mathcal{M}(t)-d\Lambda(t),
\end{eqnarray}
where
\begin{eqnarray*}
\mathcal{M}(t)\equiv\int_{0}^{t}e^{-\tilde{r}u}\tilde{S}_{u}f^{'}(\tilde{S}_{u})\sigma
d\mathcal{W}(u)
\end{eqnarray*}
is a  martingale, and
\begin{eqnarray*}
\Lambda(t)\equiv\int_{0}^{t}e^{-\tilde{r}u}
[(\delta\tilde{S}_{u}-\tilde{r}q)I_{\{\tilde{S}_{u}>b\}}]du
\end{eqnarray*}
is  a nonnegative and nondecreasing process because $\delta
x-\tilde{r}q\geq0,x>b$ with $b>q\geq\frac{r-\gamma}{\delta}q$ under
the assumption $\delta>0$ and $\gamma-r+\delta\geq0$, and
$\tilde{r}= r-\gamma<0$ under the assumption $\delta=0$ and
$\gamma-r>\frac{\sigma^{2}}{2}$.

For any stopping time $\tau$ and any $t\in [0,\infty)$, by
(\ref{variational}), (\ref{Ito1}) and Proposition \ref{inequality}
we have
\begin{eqnarray}\label{fenjie}
f(\tilde{S}_{0})&=&{\bf E}\big [e^{-\tilde{r}(\tau\wedge
\tau_{a}\wedge t)}f(\tilde{S}_{\tau\wedge \tau_{a}\wedge t})\big
]+{\bf E}\big [\Lambda(\tau\wedge
\tau_{a}\wedge t)\big ]\nonumber\\
&\geq &{\bf E}\big [e^{-\tilde{r}(\tau\wedge \tau_{a}\wedge
t)}f(\tilde{S}_{\tau\wedge \tau_{a}\wedge t})\big ]\nonumber\\
&=&{\bf E}\big [e^{-\tilde{r}(\tau\wedge t)}f(\tilde{S}_{\tau\wedge
t})I_{\{\tau<\tau_{a}\}}\big ]+{\bf
E}\big[e^{-\tilde{r}(\tau_{a}\wedge
t)}f(\tilde{S}_{\tau_{a}\wedge t})I_{\{\tau_{a}\leq\tau\}}\big ]\nonumber\\
&\geq&{\bf E}\big [e^{-\tilde{r}(\tau\wedge
t)}(\tilde{S}_{\tau\wedge t}-q)_{+}I_{\{\tau<\tau_{a}\}}\big ]+{\bf
E}\big [e^{-\tilde{r}(\tau_{a}\wedge t)}f(\tilde{S}_{\tau_{a}\wedge
t})I_{\{\tau_{a}\leq\tau\}}\big ]\nonumber\\
&=&{\bf E}\big [e^{-\tilde{r}(\tau\wedge t)}(\tilde{S}_{\tau\wedge
t}-q)_{+}I_{\{\tau<\tau_{a}\}}\big ]+{\bf E}\big
[e^{-\tilde{r}t}f(\tilde{S}_{t})I_{\{\tau_{a}\leq\tau\}}I_{\{\tau_{a}>
t\}}\big ],
\end{eqnarray}
where we have used $f(\tilde{S}_{\tau_{a}})=0$.

Obviously,
\begin{eqnarray*} e^{-\tilde{r}(\tau\wedge
t)}(\tilde{S}_{\tau\wedge
t}-q)_{+}I_{\{\tau<\tau_{a}\}}\leq\sup\limits_{0\leq
t<\infty}e^{-\tilde{r}t}(\tilde{S}_{t}-q)_{+}
\end{eqnarray*}
and
\begin{eqnarray*}
e^{-\tilde{r}(\tau_{a}\wedge t)}f(\tilde{S}_{\tau_{a}\wedge
t})I_{\{\tau_{a}\leq\tau\}}\leq\sup\limits_{0\leq
t<\infty}e^{-\tilde {r}t}\tilde{S}_{t}.
\end{eqnarray*}
By Lemma 3.1 in \cite{stock} we have
\begin{eqnarray}
{\bf E}\big[\sup\limits_{0\leq
t<\infty}e^{-\tilde{r}t}(\tilde{S}_{t}-q)_{+}\big]<\infty
\end{eqnarray}
if $\delta>0$ and $\gamma-r+\delta\geq 0$ or $\delta=0$ and
$\gamma-r>\frac{\sigma^{2}}{2}$. By using the
 dominated convergence theorem  and letting $t\rightarrow \infty$
\begin{eqnarray}\label{convergence1}
{\bf E}\big [e^{-\tilde{r}(\tau\wedge t)}(\tilde{S}_{\tau\wedge
t}-q)_{+}I_{\{\tau<\tau_{a}\}}\big ]\rightarrow {\bf E}\big
[e^{-\tilde{r}\tau}(\tilde{S}_{\tau}-q)_{+}I_{\{\tau<\tau_{a}\}}].
\end{eqnarray}
In order for (\ref{step1}), we claim that the second term on the
right-side of (\ref{fenjie}) tends to 0 as $t\rightarrow \infty$. By
Proposition \ref{inequality} and H\"{o}lder's inequality
\begin{eqnarray}
{\bf E}\big
[e^{-\tilde{r}t}f(\tilde{S}_{t})I_{\{\tau_{a}\leq\tau\}}I_{\{\tau_{a}>
t\}}\big ]&\leq&{\bf E}\big
[e^{-\tilde{r}t}\tilde{S}_{t}I_{\{\tau_{a}> t\}}\big ]\nonumber\\
&\leq&\big[{\bf
E}(e^{-\tilde{r}t}\tilde{S}_{t})^{1+\epsilon}\big]^{\frac{1}{1+\epsilon}}\big[{\bf
E}(I_{\{\tau_{a}>
t\}})\big]^{\frac{\epsilon}{1+\epsilon}},\epsilon>0.\nonumber\\
\end{eqnarray}
It is easy to derive
\begin{eqnarray}\label{holder1}
\big [{\bf
E}(e^{-\tilde{r}t}\tilde{S}_{t})^{1+\epsilon}\big]^{\frac{1}{1+\epsilon}}=
S_{0}e^{-\delta t+\frac{\epsilon\sigma^{2}}{2}t}.
\end{eqnarray}
Next we prove that $[{\bf E}(I_{\{\tau_{a}>
t\}})\big]^{\frac{\epsilon}{1+\epsilon}}\leq \alpha
e^{-\frac{\mu^{2}\epsilon}{2(1+\epsilon)}t}$. Since
\begin{eqnarray*}
\tau_{a}=\tau_{a_{1}}=\inf{\{t\geq 0:\mathcal{W}_{t}+\mu t\leq
a_{1}\}},
\end{eqnarray*}
where $\mu=-(\frac{\sigma}{2}+\frac{\gamma-r+\delta}{\sigma}),
a_{1}=\frac{1}{\sigma}\log\frac{a}{S_{0}}$, using  density of
hitting time $\tau_{a_{1}}$(cf.\cite{Handbook}) we have
\begin{eqnarray*}
{\bf E}(I_{\{\tau_{a}>
t\}})&=&\int_{t}^{\infty}\frac{|a_{1}|}{\sqrt{2\pi
u^{3}}}e^{-\frac{(\mu u-a_{2})^{2}}{2u}}du\\
&=&\int_{t}^{\infty}\frac{|a_{1}|}{\sqrt{2\pi
u^{3}}}e^{-\frac{\mu^{2}u}{2}+\mu a_{1}-\frac{a_{1}^{2}}{2u}}du\\
&\leq&\alpha_{1}\int_{t}^{\infty}e^{-\frac{\mu^{2}u}{2}}du,\\
&\leq&\alpha_{2}e^{-\frac{\mu^{2}}{2}t}
\end{eqnarray*}
for $t$ sufficiently large, where $\alpha_{1}$ and $\alpha_{2}$ are
some positive  constants,  so
\begin{eqnarray}\label{holder2}
[{\bf E}(I_{\{\tau_{a}> t\}})\big]^{\frac{\epsilon}{1+\epsilon}}\leq
\alpha e^{-\frac{\mu^{2}\epsilon}{2(1+\epsilon)}t},
\end{eqnarray}
and  $\alpha>0$ is a constant. Because  we can find $\epsilon>0$
such that
$\delta-\frac{\epsilon\sigma^{2}}{2}+\frac{\mu^{2}\epsilon}{2(1+\epsilon)}>0$
if $\delta>0$, or  $\delta=0$ and $\gamma-r>\frac{\sigma^{2}}{2}$,
by
 (\ref{holder1}) and (\ref{holder2}) we have
\begin{eqnarray}\label{convergence3}
{\bf E}\big
[e^{-\tilde{r}t}f(\tilde{S}_{t})I_{\{\tau_{a}\leq\tau\}}I_{\{\tau_{a}>
t\}}\big ]&\leq&{\bf E}\big
[e^{-\tilde{r}t}\tilde{S}_{t}I_{\{\tau_{a}> t\}}\big ]\nonumber\\
&\leq&\big[{\bf
E}(e^{-\tilde{r}t}\tilde{S}_{t})^{1+\epsilon}\big]^{\frac{1}{1+\epsilon}}\big[{\bf
E}(I_{\{\tau_{a}> t\}})\big]^{\frac{\epsilon}{1+\epsilon}}\nonumber\\
&\leq&S_{0}e^{-\delta t+\frac{\epsilon\sigma^{2}}{2}t}\alpha
e^{-\frac{\mu^{2}\epsilon}{2(1+\epsilon)}t}\nonumber\\
&=&\alpha
S_{0}e^{-(\delta-\frac{\epsilon\sigma^{2}}{2}+
\frac{\mu^{2}\epsilon}{2(1+\epsilon)})t}\rightarrow
0,t\rightarrow \infty.
\end{eqnarray}

Using (\ref{convergence1}), (\ref{convergence3}) and letting
$t\rightarrow \infty$ in (\ref{fenjie}),
\begin{eqnarray}
f(\tilde{S}_{0})&\geq&{\bf E}\big
[e^{-\tilde{r}\tau}(\tilde{S}_{\tau}-q)_{+}I_{\{\tau<\tau_{a}\}}\big
].
\end{eqnarray}
\vskip 10pt {\bf Step two.} We show that
\begin{eqnarray}
 f(x)={\bf E}\big
[e^{-\tilde{r}\tau_{b}}(\tilde{S}_{\tau_{b}}-q)_{+}I_{\{\tau_{b}<\tau_{a}\}}\big
].
\end{eqnarray}
Let $\tau=\tau_{b}$, we have $\Lambda(\tau_{b}\wedge\tau_{a})=0$,
$f(\tilde{S}_{\tau_{b}})=\tilde{S}_{\tau_{b}}-q$ and
$f(\tilde{S}_{\tau_{a}})=0$, hence the (\ref{fenjie}) becomes
\begin{eqnarray*}
f(\tilde{S}_{0})={\bf E}\big [e^{-\tilde{r}\tau_{b}}(
\tilde{S}_{\tau_{b}}-q)_{+}I_{\{\tau_{b}<\tau_{a},\tau_{b}\leq
t\}}\big ]+{\bf E}\big
[e^{-\tilde{r}t}f(\tilde{S}_{t})I_{\{t<\tau_{b},t< \tau_{a}\}}\big
].
\end{eqnarray*}
By (\ref{convergence3})
\begin{eqnarray*}
{\bf E}\big [e^{-\tilde{r}t}f(\tilde{S}_{t})I_{\{t< \tau_{b}, t<
\tau_{a}\}}\big ]\rightarrow 0,t\rightarrow \infty.
\end{eqnarray*}
Then
\begin{eqnarray}
f(\tilde{S}_{0})={\bf E}\big
[e^{-\tilde{r}\tau_{b}}(\tilde{S}_{\tau_{b}}-q)_{+}I_{\{\tau_{b}<\tau_{a}\}}\big
].
\end{eqnarray}
Thus we complete the proof.
\end{proof}
\vskip10pt
 \noindent
\begin{Remark}
Given an initial stock price $S_{0}=x$, $\tau_{a}$ exists and is
determined by the bank and the client initially. By Theorem
\ref{main} $\tau_{b}$ is the optimal stopping time, the client will
regain the stock at $\tau_{b}$ to get maximum return by paying
amount $qe^{\gamma \tau_{b}}  $ to the bank before the stock loan is
terminated. So the stock loan is terminated at stopping time
$\tau_{a}\wedge \tau_{b}$.
\end{Remark}
\vskip10pt  \noindent
\begin{Remark} By the same procedure as in  the initial value $f(x)$,
we can easily get
\begin{eqnarray*}
e^{-\tilde{r}t}f(\tilde{S}_{t})&=&\sup\limits_{\tau \in \mathcal
{T}_{t}}{\bf E}\big[e^{-\tilde{r}\tau}(\tilde
{S}_{\tau}-q)_{+}I_{\{\tau<\tau_{a}\}}|\mathcal{F}_{t}\big]\\
&=&e^{-rt}V_{t}
\end{eqnarray*} and
\begin{eqnarray*}
V_{t}=e^{\gamma t}f(e^{-\gamma t}S_{t}).
\end{eqnarray*}
\end{Remark}
\vskip 15pt  \noindent Now we calculate $f(x)$ via using Theorem
\ref{main}.
 We only need to work out $f(x)$ in the region $(a,b)$ by smooth
fit principle. For this, it suffices to  solve the following
problem,
\begin{eqnarray}\label{equation_ab}
\left\{
\begin{array}{l l}
\frac{1}{2}\sigma^{2}x^{2}f^{''}+(\tilde{r}-\delta)xf^{'}-\tilde{r}f=0,\ a<x<b,\\
f(a)=0,f(b)=b-q,f^{'}(b)=1.
\end{array}
\right.
\end{eqnarray}
The general solutions of (\ref{equation_ab}) has the following form,
\begin{eqnarray*}
f(x)=C_{1}x^{\lambda_{1}}+C_{2}x^{\lambda_{2}}
\end{eqnarray*}
and the $\lambda_{1}$ and $\lambda_{2}$ are defined by
\begin{eqnarray}\label{solve_lumda}
\lambda_{1}=\frac{-\mu+\sqrt{\mu^{2}-2(\gamma-r)}} {\sigma},\
\lambda_{2}=\frac{-\mu-\sqrt{\mu^{2}-2(\gamma-r)}}{\sigma},
\end{eqnarray}
where $\mu=-(\frac{\sigma}{2}+\frac{\gamma-r+\delta}{\sigma})$.

If $\delta>0$ and $\gamma-r+\delta\geq0$, then
$\lambda_{1}>1>\lambda_{2}$. If $\delta=0$ and
$\gamma-r>\frac{\sigma^{2}}{2}$, then
$\lambda_{1}=\frac{2(\gamma-r)}{\sigma^{2}}>1=\lambda_{2}$.

By the boundary conditions we have
\begin{eqnarray}\label{ysolution}
\left\{
\begin{array}{l l l}
f(a)=C_{1}a^{\lambda_{1}}+C_{2}a^{\lambda_{2}}=0,\\
f(b)=C_{1}b^{\lambda_{1}}+C_{2}b^{\lambda_{2}}=b-q,\\
f^{'}(b)=C_{1}\lambda_{1}b^{\lambda_{1}-1}+C_{2}\lambda_{2}b^{\lambda_{2}-1}=1.
\end{array}
\right.
\end{eqnarray}
Solving  the first two equations of (\ref{ysolution}) we obtain
$C_{2}=-C_{1}a^{\lambda_{1}-\lambda_{2}}$ and
$C_{1}=\frac{b-q}{b^{\lambda_{1}}-a^{\lambda_{1}-\lambda_{2}}b^{\lambda_{2}}}$.
By the last equality in (\ref{ysolution}) and  letting $b=a y$ we
have
\begin{eqnarray}\label{sovle_b}
g(y)&
\equiv&(\lambda_{1}-1)y^{\lambda_{1}+1}-\frac{q}{a}\lambda_{1}y^{\lambda_{1}}+
(1-\lambda_{2})y^{\lambda_{2}+1}+\frac{q}{a}\lambda_{2}y^{\lambda_{2}}\nonumber\\
&=&0.
\end{eqnarray}
If $y^{*}$ solves the equation (\ref{sovle_b}), then $b=ay^{*}$. $b$
only depends on $a$  for fixed $(\gamma,\delta,\sigma,q)$. Thus
\begin{eqnarray*}
C_{1}=\frac{1}{C}(b-q)b^{\frac{\mu}{\sigma}}
a^{-\frac{\sqrt{\mu^{2}-2\lambda}}{\sigma}}
\end{eqnarray*}
and
\begin{eqnarray*}
C_{2}=-\frac{1}{C}(b-q)b^{\frac{\mu}{\sigma}}
a^{\frac{\sqrt{\mu^{2}-2\lambda}}{\sigma}},
\end{eqnarray*}
where $C=(\frac{b}{a})^{\frac{\sqrt{\mu^{2}-2\lambda}}{\sigma}}
-(\frac{a}{b})^{\frac{\sqrt{\mu^{2}-2\lambda}}{\sigma}}$.  We will
show that the $ y^*$ determined by (3.7) is unique and $b=ay^{*}$
exists in next section.
\begin{Remark}
Dai and Xu\cite{Dai1}(2010) solved other stock loan by variation
approach. It seems that the the proof in \cite{Dai1} does not work
for Theorem\ref{main} because of the automatic termination clause.
The proof of Theorem\ref{main} needs delicate estimates.
\end{Remark}
\vskip 15pt
\noindent
 \setcounter{equation}{0}
\section{{\small {\bf Probabilistic Solution}}}
\vskip 5pt \noindent In this section we will give  the probabilistic
solution of stock loan with automatic termination clause. The
initial stock price $S_{0}=x$. Using Theorem \ref{main}, $\tau_{b}$
is the optimal stopping time and $\{\tau_{a}=\tau_{b}\}=\baro$ for
$a\neq b$, it is easy to see from (\ref{reward}) that
\begin{eqnarray}\label{resultexp}
f(x)&=&{\bf E}\big [e^{-\tilde{r}\tau_{b}}(\tilde
{S}_{\tau_{b}}-q)_{+}I_{\{\tau_{b}<\tau_{a}\}}\big ].
\end{eqnarray}
 Therefore  we have the
following.
\vskip10pt
 \noindent
\begin{CR}\label{expectform} We assume the same conditions as in
Theorem \ref{main}. Then
\begin{eqnarray} f(x)=\left\{
\begin{array}{l l l}
0,& x\leq a,\\
x-q,& x\geq b,\\
(b-q){\bf E}\big
[e^{-\tilde{r}\tau_{b}}I_{\{\tau_{b}<\tau_{a}\}}\big ],&a<x<b.
\end{array}
\right.
\end{eqnarray}
\end{CR}
\noindent Now we compute the following expectation with the initial
price $x=S_{0}$ in the interval $(a,b)$,
\begin{eqnarray}\label{expecta}
{\bf E}\big [e^{-\tilde{r}\tau_{b}}I_{\{\tau_{b}<\tau_{a}\}}\big ].
\end{eqnarray}
Define
\begin{eqnarray*}
&&\mu=-(\frac{\sigma}{2}+\frac{\gamma-r+\delta}{\sigma}),\ \  \lambda=\gamma-r,\\
&& b_{1}=\frac{1}{\sigma}\log\frac{b}{S_{0}}, \ \
a_{1}=\frac{1}{\sigma}\log\frac{a}{S_{0}}.
\end{eqnarray*}
Obviously,
\begin{eqnarray}
\tau_{a}=\tau_{a_{1}}=\inf{\{t\geq 0:\mathcal{W}_{t}+\mu t\leq
a_{1}\}}
\end{eqnarray}
and
\begin{eqnarray}
\tau_{b}=\tau_{b_{1}}=\inf{\{t\geq 0:\mathcal {W}_{t}+\mu t\geq
b_{1}\}}.
\end{eqnarray}
Using well-known results about standard Brownian motion on an
interval and Girsanov theorem (cf.\cite{Brownian}), we compute
(\ref{expecta}) as the following. \vskip10pt
 \noindent
\begin{Lemma}\label{twoexpet}
If $\mu^{2}-2\lambda\geq 0$, then
\begin{eqnarray}\label{expect2}
{\bf E}\big [e^{-\tilde{r}\tau_{b}}I_{\{\tau_{b}<\tau_{a}\}}\big ]
&=&
{\bf E}\big [e^{-\tilde{r}\tau_{b_{1}}}I_{\{\tau_{b_{1}}<\tau_{a_{1}}\}}\big]\nonumber\\
&=&\frac{1}{C}(e^{\mu b_{1}-a_{1}\sqrt{\mu^{2}-2\lambda}}-e^{\mu
b_{1}+a_{1}\sqrt{\mu^{2}-2\lambda}})\nonumber\\
&=&\frac{1}{C}(b^{\frac{\mu}{\sigma}}a^{-\frac{\sqrt{\mu^2-2\lambda}}{\sigma}}x^{\lambda_{1}}
-b^{\frac{\mu}{\sigma}}a^{\frac{\sqrt{\mu^2-2\lambda}}{\sigma}}x^{\lambda_{2}})
\end{eqnarray}
where $C=(\frac{b}{a})^{\frac{\sqrt{\mu^{2}-2\lambda}}{\sigma}}
-(\frac{a}{b})^{\frac{\sqrt{\mu^{2}-2\lambda}}{\sigma}}$, $x=S_{0}$
and $\lambda=\gamma-r$.
\end{Lemma}
\begin{proof} It is well known (cf.\cite{Handbook,Brownian}) that
the density of $\tau_{b_{1}}$ under $\tau_{b_{1}}<\tau_{a_{1}}$ is
\begin{eqnarray*}
P(\tau_{b_{1}}\in dt,\tau_{b_{1}}<\tau_{a_{1}})=\frac{e^{\mu
b_{1}-\frac{1}{2}\mu^{2}t}}{\sqrt{2\pi
t^{3}}}\sum_{n=-\infty}^{+\infty}(2n(b_{1}-a_{1})+b_{1})
e^{-\frac{(2n(b_{1}-a_{1})+b_{1})^{2}}{2t}}dt.
\end{eqnarray*}
If $\mu^{2}-2\lambda \geq 0$,  then,  by Laplace transform of the
law of hitting time of Brownian motion with drift, it easily follows
that (cf.\cite{Handbook,Brownian,stock})
\begin{eqnarray}\label{fubini}
{\bf E}\big (e^{\lambda\tau_{b}}I_{\{\tau_{b}<\tau_{a}\}}\big )&=&
{\bf E}\big (e^{\lambda\tau_{b_{1}}}I_{\{\tau_{b_{1}}<\tau_{a_{1}}\}}\big )\nonumber\\
&=&\int_{0}^{+\infty}e^{\lambda t}P(\tau_{b_{1}}\in
dt,\tau_{b_{1}}<\tau_{a_{1}})\nonumber\\
&=&\int_{0}^{+\infty}e^{\lambda t}\frac{e^{\mu
b_{1}-\frac{1}{2}\mu^{2}t}}{\sqrt{2\pi
t^{3}}}\sum_{n=-\infty}^{+\infty}(2n(b_{1}-a_{1})
+b_{1})e^{-\frac{(2n(b_{a}-a_{1})+b_{1})^{2}}{2t}}dt\nonumber\\
&=&e^{\mu
b_{1}}\sum_{n=-\infty}^{+\infty}\int_{0}^{+\infty}e^{\lambda
t}e^{-\frac{1}{2}\mu^{2}t}\frac{1}{\sqrt{2\pi
t^{3}}}(2n(b_{1}-a_{1})+b_{1})e^{-\frac{(2n(b_{a}-a_{1})+b_{1})^{2}}{2t}}dt\nonumber\\
&=&e^{\mu a_{1}-\mu
\tilde{x}}\sum_{n=-\infty}^{+\infty}\int_{0}^{+\infty}e^{\lambda
t}\frac{1}{\sqrt{2\pi t^{3}}}\tilde{x}e^{-\frac{(\tilde{x}-\mu
t)^{2}}{2t}}dt,
\end{eqnarray}
where $\tilde{x}=2n(b_{1}-a_{1})+b_{1}$, if $n\geq 0,\tilde{x}\geq
0$; otherwise $\tilde{x}<0$. The fourth equality follows from
Fubini's theorem.

If $\mu^{2}-2\lambda>0$, then we can choose $\varepsilon>0$ such
that $\mu^{2}-2(\lambda+\varepsilon)>0$. We first consider the case:
$n\geq 0,\tilde{x}>0$,
\begin{eqnarray}\label{positive}
\int_{0}^{+\infty}e^{\lambda t}\frac{1}{\sqrt{2\pi
t^{3}}}\tilde{x}e^{-\frac{(\tilde{x}-\mu
t)^{2}}{2t}}dt&=&\int_{0}^{+\infty}e^{\lambda t}\frac{1}{\sqrt{2\pi
t^{3}}}|\tilde{x}|e^{-\frac{(\tilde{x}-\mu t)^{2}}{2t}}dt\nonumber\\
&=&e^{-\tilde{x}(\sqrt{\mu^{2}-2(\lambda+\varepsilon)}-\mu)}
\int_{0}^{+\infty}\frac{|\tilde{x}|}{\sqrt{2\pi
t^{3}}}e^{-\varepsilon
t}e^{-\frac{(\tilde{x}-\sqrt{\mu^{2}-2
(\lambda+\varepsilon)}t)^{2}}{2t}}dt\nonumber\\
&=&e^{-\tilde{x}(\sqrt{\mu^{2}-2(\lambda+\varepsilon)}-\mu)}e^{\sqrt{\mu^{2}
-2(\lambda+\varepsilon)}\tilde{x}-|\tilde{x}|\sqrt{\mu^{2}-
2(\lambda+\varepsilon)+2\varepsilon}}\nonumber\\
&=&e^{\mu\tilde{x}-|\tilde{x}|\sqrt{\mu^{2}-2\lambda}}.
\end{eqnarray}
Similarly, for  $n\leq -1,\tilde{x}<0$,
\begin{eqnarray}\label{negative}
\int_{0}^{+\infty}e^{\lambda t}\frac{1}{\sqrt{2\pi
t^{3}}}\tilde{x}e^{-\frac{(\tilde{x}-\mu
t)^{2}}{2t}}dt=-e^{\mu\tilde{x}-|\tilde{x}|\sqrt{\mu^{2}-2\lambda}}.
\end{eqnarray}
Hence, by  (\ref{fubini}),(\ref{positive}) and (\ref{negative})
\begin{eqnarray}
{\bf E}\big (e^{\lambda\tau_{b}}I_{\{\tau_{b}<\tau_{a}\}}\big )
&=&{\bf E}\big (e^{\lambda\tau_{b_{1}}}
I_{\{\tau_{b_{1}}<\tau_{a_{1}}\}}\big )\nonumber\\
&=&e^{\mu b_{1}}\sum_{n=0}^{\infty}e^{-\mu
\tilde{x}}e^{\mu\tilde{x}-\tilde{x}\sqrt{\mu^{2}-2\lambda}}-e^{\mu
b_{1}}\sum_{n=-1}^{-\infty}e^{-\mu
\tilde{x}}e^{\mu\tilde{x}+\tilde{x}\sqrt{\mu^{2}-2\lambda}}\nonumber\\
&=&e^{\mu
b_{1}}(\sum_{n=0}^{\infty}e^{-\tilde{x}\sqrt{\mu^{2}-2\lambda}}
-\sum_{n=-1}^{-\infty}e^{\tilde{x}\sqrt{\mu^{2}-2\lambda}})\nonumber\\
&=& \frac{1}{C}(e^{\mu b_{1}-a_{1}\sqrt{\mu^{2}-2\lambda}}-e^{\mu
b_{1}+a_{1}\sqrt{\mu^{2}-2\lambda}})\nonumber\\
&=&\frac{1}{C}\big
(b^{\frac{\mu}{\sigma}}a^{\frac{-\sqrt{\mu^{2}-2\lambda}}{\sigma}}x^{\lambda_{1}}
-b^{\frac{\mu}{\sigma}}a^{\frac{\sqrt{\mu^{2}-2\lambda}}{\sigma}}x^{\lambda_{2}}\big
).
\end{eqnarray}
For $\mu^{2}-2\lambda=0$, the conclusion follows from
$\lambda_{n}\uparrow\lambda$ and  monotone convergence theorem. Thus
we complete the proof.
\end{proof}
 \noindent By Corollary \ref{expectform} and Lemma \ref{expectform} we
have
\begin{eqnarray}\label{solution}
f(x)=\left\{
\begin{array}{l l l}
0,& x\leq a,\\
x-q,& x\geq b,\\
\frac{b-q}{C}(b^{\frac{\mu}{\sigma}}
a^{\frac{-\sqrt{\mu^{2}-2\lambda}}{\sigma}}x^{\lambda_{1}}
-b^{\frac{\mu}{\sigma}}
a^{\frac{\sqrt{\mu^{2}-2\lambda}}{\sigma}}x^{\lambda_{2}}),&a<x<b,
\end{array}
\right.
\end{eqnarray}
where $S_{0}=\tilde{S}_{0}=x$, $C$ is given in Lemma
\ref{expectform},  $\lambda_{1}$ and $\lambda_{2}$  are given by
(\ref{solve_lumda}). It is easy to check that the above solution is
the same  solution as in last section.  $f(x)$ is continuous and
second order continuously differentiable except  points $a$ and $b$.
It suffices to compute $b$ in order to show that $f$ satisfies the
assumption in Theorem \ref{main}, that is, $f(x)$ is first order
continuously differentiable at the point $b$. \vskip 15pt\noindent
\begin{Remark} The proof of Lemma \ref{twoexpet} is somewhat similar to those in
Xia and Zhou \cite{stock}. Our case is more complicate and is very
difficulty in computation of (\ref{expect2}) and Theorem
\ref{wuweiming5} below.
\end{Remark}
\vskip15pt
 \noindent
Let $f^{'}(b)=1$, we want to show that there exists $
y^{*}>\frac{q}{a}$ satisfying (\ref{sovle_b}) and $y^{*}$ is unique
under certain assumptions on the parameters
$\gamma,r,\delta,\sigma,a$. \vskip10pt
 \noindent
\begin{Prop}\label{solution_exist}
If $\delta>0$ and $\gamma-r+\delta \geq0$, then there exists $
y^{*}>\frac{q}{a}$ such that $g(y^{*})=0$ and the $y^{*}$ is unique.
 $b=ay^{*}>q$ is unique too, where
$h(y)=\frac{\lambda_{1}+1-\lambda_{2}}
{\lambda_{1}}y^{1-\lambda_{2}}-\frac{q}{a}
\frac{\lambda_{1}-1-\lambda_{2}}{\lambda_{1}-1}y^{-\lambda_{2}}$,
$g(y)$ is defined by (\ref{sovle_b}).
\end{Prop}
\begin{proof}  Since $\delta>0$, we have $\lambda_{1}>1>\lambda_{2}$,
\begin{eqnarray*}
g(\frac{q}{a})=(\frac{q}{a})^{\lambda_{2}+1}
(1-(\frac{q}{a})^{\lambda_{1}-\lambda_{2}})<0
\end{eqnarray*}
and
\begin{eqnarray*}
\lim\limits_{y\rightarrow \infty}g(y)=\infty.
\end{eqnarray*}
By  continuity of $g(y)$, there exists $y^{*}>\frac{q}{a}$ such that
$g(y^*)=0$ and $b=ay^{*}>q$. Moreover,  it is easy to see from the
procedure in section 3  that the assumptions in Theorem 3.1 hold
for the $b$.\\
Next we prove the uniqueness of $y^*$. Define
\begin{eqnarray*}
\tilde{g}(y)&=&y^{-\lambda_{2}}g(y)\\
&=&(\lambda_{1}-1)y^{\lambda_{1}+1-\lambda_{2}}-
\frac{q}{a}\lambda_{1}y^{\lambda_{1}-\lambda_{2}}+
(1-\lambda_{2})y+\frac{q}{a}\lambda_{2}.
\end{eqnarray*}
Then
\begin{eqnarray*}
\tilde{g}^{''}(y)
&=&(\lambda_{1}-1)(\lambda_{1}+1-\lambda_{2})
(\lambda_{1}-\lambda_{2})y^{\lambda_{1}-\lambda_{2}-1}\\
&&-\frac{q}{a}\lambda_{1}(\lambda_{1}-\lambda_{2})
(\lambda_{1}-\lambda_{2}-1)y^{\lambda_{1}-\lambda_{2}-2}.
\end{eqnarray*}
Since $\tilde{g}^{''}(y)\geq0$, $\tilde{g}^(y)$ is convex (see lemma
6.1 in the appendix). So the uniqueness of $y^{*}$ easily follows
from the convexity and $\tilde{g}^(\frac{q}{a})< 0 $. Thus we
complete the proof.
\end{proof}
\vskip10pt
 \noindent
\begin{Remark}
The convexity of function $\tilde{g}(y)$ will be given in detail  in
Lemma \ref{AL} below.
\end{Remark}
\vskip5pt
 \noindent
\begin{Prop}\label{solution_exist2}
If $\delta=0 $ and $\gamma-r>\frac{\sigma^{2}}{2}$, then there
exists $ y^{*}>\frac{q}{a}$ such that $g(y^{*})=0$ and the $y^{*}$
is unique. So $b=ay^{*}>q$ is unique too, where $g(y)$ is defined by
(\ref{sovle_b}).
\end{Prop}
\begin{proof} Since $\delta=0$ and $\gamma-r>\frac{\sigma^{2}}{2}$, we have
$\lambda_{1}=\frac{2(\gamma-r)}{\sigma^{2}}>1=\lambda_{2}$. It is
easy to prove $\tilde{g}^{''}(y)\geq0$. By an argument similar to
the proof of Proposition \ref{solution_exist},We can complete the
proof.
\end{proof}
\vskip10pt
 \noindent
\begin{Remark}
$\tau_{a}$ is the automated terminable stopping time of the stock
loan. The automatic termination clause provides a  protection for
the bank. However, the client may have more or less motivation to
take risk compared to the circumstance without the clause (or $a=0$)
via the value of $a$. Denote $\tau_{b(a)}$ is the optimal stopping
time and $f_{a}(x)$ is the initial value with the automatic
termination clause. Intuitively, we have
\begin{eqnarray*}
\lim_{a\rightarrow 0+}{b(a)}=b(0)
\end{eqnarray*}
and
\begin{eqnarray*}
\lim_{a\rightarrow 0+}{f_{a}(x)}=f_{0}(x),
\end{eqnarray*}
where $\tau_{b_(0)}$ is the optimal stopping time and $f_{0}(x)$ is
the initial value without the automatic termination clause
introduced by Xia and Zhou \cite{stock}. The consistent result
follows  from Proposition \ref{Prop4.3} below  in the case where
$\delta>0$ and $\gamma-r+\delta\geq0$.
\end{Remark}
\vskip10pt
 \noindent
\begin{Prop} \label{Prop4.3}
Assume that $\delta>0,\gamma-r+\delta\geq0$ and
$\delta=0,\gamma-r>\frac{\sigma^{2}}{2}$. Then
 we have\\
(1)\  $\lim\limits_{a\rightarrow 0+}b(a)=b(0)$.\\
(2)\  $\lim\limits_{a\rightarrow0+}f_{a}(x)= f_{0}(x)=\left\{
\begin{array}{l l}
x-q,& x\geq b(0),\\
(b(0)-q)(\frac{x}{b(0)})^{\lambda_{1}},&x<b(0),
\end{array}
\right. $\\
where $b(0)=\frac{q\lambda_{1}}{\lambda_{1}-1}$, $\lambda_{1}$ is
given by (\ref{solve_lumda}).
\end{Prop}
\begin{proof} \  We first prove(1). By (\ref{sovle_b}) and
$y=\frac{b}{a}$
\begin{eqnarray*}
F(a,b)&=&a^{\lambda_{1}+1}g(\frac{b}{a})\\
&=&(\lambda_{1}-1)b^{\lambda_{1}+1}
-\lambda_{1}qb^{\lambda_{1}}-(\lambda_{2}-1)b^{\lambda_{2}
+1}a^{\lambda_{1}-\lambda_{2}}
+\lambda_{2}qb^{\lambda_{2}}a^{\lambda_{1}-\lambda_{2}}.
\end{eqnarray*}
Since $F(a,b)$ and $F'_{b}(a,b)$ are continuous on
$[0,q)\times[q,\infty)$, $F(0,b(0))=0$ and $F_{b}(0,b(0))>0$, by
implicit function theorem, there exists $\rho>0$ such that $b$ is an
function of $a$ in the region $[0,\rho)$ and $b(a)$ is continuous.
Thus $\lim\limits_{a\rightarrow 0+}b(a)=b(0)$.

Next we turn to proving (2). Since $\lambda_{1}>1\geq \lambda_{2}$,
by using (\ref{solution}) we have
\begin{eqnarray*}
\lim_{a\rightarrow0+}f_{a}(x)=f_{0}(x)=\left\{
\begin{array}{l l}
x-q,& x\geq b(0),\\
(b(0)-q)(\frac{x}{b(0)})^{\lambda_{1}},&x<b(0).
\end{array}
\right.
\end{eqnarray*}
Therefore we complete the proof.
\end{proof}
\vskip10pt
 \noindent
\begin{Remark}
Proposition \ref{Prop4.3} shows that the stock loan with automatic
termination clause is consistent with the result given by Xia and
Zhou in \cite{stock} as $a\rightarrow0+$.
\end{Remark}
\vskip 5pt \noindent As a direct consequence of (\ref{solution}),
Propositions \ref{solution_exist}-\ref{solution_exist2} and Theorem
3.1, we get the initial value $ f(S_0)$ of the stock loan with
automatic termination clause as follows. \vskip10pt
 \noindent
\begin{Them}\label{mainresult}
Assume that $\delta>0$ and $\gamma-r+\delta\geq0$ or $\delta=0$ and
$\gamma-r>\frac{\sigma^{2}}{2}$. Define $f$ by (\ref{solution}), $b$
by Proposition \ref{solution_exist} and Proposition
\ref{solution_exist2}. Then the initial value of stock loan with
automatic termination clause is $ f(S_0)$.
\end{Them}
\vskip 5pt  \noindent
\setcounter{equation}{0}
\section{{\small {\bf Stock
loan with automatic termination clause, cap and margin}}} \vskip
5pt\noindent In this section we add a cap and a margin to stock loan
with automatic termination clause to protect the lender from a large
drop in value, or even default, of the collateral. We will give
explicit formulas for the value function and the optimal exercise
time.\vskip 5pt\noindent Let the stock price S be modeled as in
(\ref{E2.1}). The value of this stock loan with automatic
termination clause, cap and margin is
\begin{eqnarray}\label{reward3}
f(x) &=&\sup\limits_{\tau \in \mathcal {T}_{0}}{\bf E}\big
[e^{-r\tau}(S_{\tau}\wedge Le^{\gamma \tau}-qe^{\gamma
\tau})_{+}I_{\{\tau<\tau_{a}\}}+ke^{-r\tau_{a}}
S_{\tau_{a}}I_{\{\tau_{a}\leq\tau\}}\big ]\nonumber\\
&=&\sup\limits_{\tau \in \mathcal {T}_{0}}{\bf E}\big
[e^{-\tilde{r}\tau}(\tilde{S}_{\tau}\wedge
L-q)_{+}I_{\{\tau<\tau_{a}\}}+ke^{-\tilde{r}
\tau_{a}}\tilde{S}_{\tau_{a}}I_{\{\tau_{a}\leq\tau\}}\big ],
\end{eqnarray}
where $\tilde{r}=r-\gamma$, $\tilde{S}_{t}=e^{-\gamma
t}S_{t},\tilde{S}_{0}=S_{0}=x$, $\mathcal {T}_{t}$ denotes all $\{
\mathcal {F}_{t}\}_{t\geq 0}$ -stopping times $\tau $ with $\tau
\geq t$ a.s., and  $\tau_{a}=\inf{\{t\geq 0,e^{-\gamma t}S_{t} \leq
a\}}$. The terms $L$ and $ kS_{\tau_{a}} $ are called $cap $ and
$margin$  satisfying $ 0<a\leq q<L $ and $ 0\leq k<1 $,
respectively. The value of this stock loan at  any time $t$ is
\begin{eqnarray}
V_{t}=\sup\limits_{\tau \in \mathcal {T}_{t}}{\bf
E}\big[e^{-r(\tau-t)}(S_{\tau}\wedge Le^{\gamma \tau}-qe^{\gamma
\tau})_{+}I_{\{\tau<\tau_{a}\}}+ke^{-r(\tau_{a}-t)}
S_{\tau_{a}}I_{\{\tau_{a}\leq\tau\}}|\mathcal{F}_{t}\big].
\end{eqnarray}
 The
contracts can be described as follows.  The stock loan has
properties as in section 2 and if the stock price falls below the
accrued loan amount, i.e., $e^{-\gamma t}S_{t} \leq a$, then the
lander pays $\theta(t)=kS_{t}$  to the borrower, and the contract is
terminated. \vskip 5pt\noindent Because solving the optimal stoping
problem (\ref{reward3}) is similar to (\ref{reward}), we omit the
details.
\begin{theorem}\label{wuweiming4}
Assume $\delta>0$ or $\delta=0,\gamma-r>\frac{\sigma^{2}}{2}$, and
the  $f(x)$ is continuous and  belongs to $ \mathcal
{C}^{1}[(0,\infty)\setminus \{a, b\wedge L\}]\cap \mathcal
{C}^{2}[(0,\infty)\setminus \{a,b\wedge L\}]$ for some $b\geq 0$.
We have the following.  \\
(1) If  $L> b$ and $f(x)$ solves the following variational
inequality
\begin{eqnarray}\label{equivalent31} \left\{
\begin{array}{l l l l}
g(x)=x\wedge L-q, & x\geq b,\\
\frac{1}{2}\sigma^{2}x^{2}g^{''}+(\tilde{r}-\delta)xg^{'}-\tilde{r}g=0, &a<x<b,\\
g(x)=kx,& x\leq a,\\
g(b)=b-q,f^{'}(b-)=1,g(a)=ka,
\end{array}
\right.
\end{eqnarray}
 then  $f(x)$ must be the  function  defined by (\ref{reward3})
  and  $\tau_{b}(=\inf{\{t\geq
0:e^{-\gamma t}S_{t}\geq b\}})\wedge \tau_{L}( =\inf{\{t\geq
0:e^{-\gamma t}S_{t}\geq L\}}) $
 is optimal in the sense that
$$f(x) ={\bf E}\big
[e^{-r\tau_{b}\wedge \tau_{L} }(S_{\tau_{b}\wedge \tau_{L} }\wedge
Le^{\gamma \tau_{b}\wedge \tau_{L} }-qe^{\gamma \tau_{b}\wedge
\tau_{L} })_{+}I_{\{\tau_{b}\wedge \tau_{L}
<\tau_{a}\}}+ke^{-r\tau_{a}}
S_{\tau_{a}}I_{\{\tau_{a}\leq\tau_{b}\wedge \tau_{L} \}}\big ].$$
(2) If  $L\leq b$ and  $f(x)$ solves the following variational
inequality
\begin{eqnarray}\label{equivalent32}
\left\{
\begin{array}{l l l l}
g(x)= L-q, & x\geq L,\\
\frac{1}{2}\sigma^{2}x^{2}g^{''}+(\tilde{r}-\delta)xg^{'}-\tilde{r}g=0, &a<x<L,\\
g(x)=kx,& x\leq a,\\
g(L)=L-q,g(a)=ka,
\end{array}
\right.
\end{eqnarray}
then $f(x)$ must be the function defined by (\ref{reward3}) and
$\tau_{L}=\inf{\{t\geq 0:e^{-\gamma t}S_{t}\geq L\}}$ is optimal in
the sense that
$$f(x) ={\bf E}\big
[e^{-r\tau_{L}}(S_{\tau_{L}}\wedge Le^{\gamma \tau_{L}}-qe^{\gamma
\tau_{L}})_{+}I_{\{\tau_{L}<\tau_{a}\}}+ke^{-r\tau_{a}}
S_{\tau_{a}}I_{\{\tau_{a}\leq\tau_{L}\}}\big ].$$
\end{theorem}
\vskip 5pt\noindent
 If  $\delta>0$ or
$\delta=0,\gamma-r>\frac{\sigma^{2}}{2}$ and $0\leq k\leq
h(\frac{q}{a})$, it is easy to see that there exists a unique $y^*$
solving the following equation
\begin{eqnarray}\label{sovle_b1}
(\lambda_{1}-1)y^{\lambda_{1}+1}-\frac{q}{a}\lambda_{1}y^{\lambda_{1}}+
(1-\lambda_{2})y^{\lambda_{2}+1}+\frac{q}{a}\lambda_{2}y^{\lambda_{2}}
-
k(\lambda_{1}-\lambda_{2})y^{\lambda_{1}+\lambda_{2}}=0,\nonumber\\
\end{eqnarray}
where $h(y)=\frac{\lambda_{1}+1-\lambda_{2}}
{\lambda_{1}}y^{1-\lambda_{2}}-\frac{q}{a}
\frac{\lambda_{1}-1-\lambda_{2}}{\lambda_{1}-1}y^{-\lambda_{2}}$.
\vskip 5pt\noindent Let  $b=ay^{*}>q$.  Solving (\ref{equivalent31})
and (\ref{equivalent32}) we get explicit expression of $g(x)$ as
following.\\
If $L\geq b $ then
\begin{eqnarray}\label{solution11}
g(x)=\left\{
\begin{array}{l l l}
kx,& x\leq a,\\
\frac{ka}{C(a,b)}(a^{\frac{\mu}{\sigma}}b
^{\frac{\sqrt{\mu^{2}-2\lambda}}{\sigma}}x^{\lambda_{2}}
-a^{\frac{\mu}{\sigma}}
b^{\frac{-\sqrt{\mu^{2}-2\lambda}}{\sigma}}x^{\lambda_{1}})+\\
\frac{b-q}{C(a,b)}(b^{\frac{\mu}{\sigma}}
a^{\frac{-\sqrt{\mu^{2}-2\lambda}}{\sigma}}x^{\lambda_{1}}
-b^{\frac{\mu}{\sigma}}a^{\frac{\sqrt{\mu^{2}-2\lambda}}
{\sigma}}x^{\lambda_{2}}),&a<x<b.\\
x-q,&  b\leq x\leq  L,\\
(L-q)(\frac{x}{L})^{\lambda_2}, & x\geq L.
\end{array}
\right.
\end{eqnarray}
If $ L< b$ then
\begin{eqnarray}\label{solution12}
g(x)=\left\{
\begin{array}{l l l}
kx,& x\leq a,\\
\frac{ka}{C(a,L)}(a^{\frac{\mu}{\sigma}}L
^{\frac{\sqrt{\mu^{2}-2\lambda}}{\sigma}}x^{\lambda_{2}}
-a^{\frac{\mu}{\sigma}}
L^{\frac{-\sqrt{\mu^{2}-2\lambda}}{\sigma}}x^{\lambda_{1}})+\\
\frac{L-q}{C(a,L)}(L^{\frac{\mu}{\sigma}}
a^{\frac{-\sqrt{\mu^{2}-2\lambda}}{\sigma}}x^{\lambda_{1}}
-L^{\frac{\mu}{\sigma}}a^{\frac{\sqrt{\mu^{2}-2\lambda}}
{\sigma}}x^{\lambda_{2}}),&a<x<L,\\
 (L-q)(\frac{x}{L})^{\lambda_2}
& x\geq L.
\end{array}
\right.
\end{eqnarray}
where $C(a,b)=(\frac{b}{a})^{\frac{\sqrt{\mu^{2}-2\lambda}}{\sigma}}
-(\frac{a}{b})^{\frac{\sqrt{\mu^{2}-2\lambda}}{\sigma}}$, $x=S_{0}$,
$\lambda=\gamma-r$, $ \lambda_1 $ and $ \lambda_2 $  are defined by
 (\ref{solve_lumda}).
 \vskip 5pt\noindent
Since the $g(x)$ above  belongs to $ \mathcal
{C}^{1}[(0,\infty)\setminus \{a, b\wedge L\}]\cap \mathcal
{C}^{2}[(0,\infty)\setminus \{a,b\wedge L\}]$ for some $b\geq 0$ and
solve (\ref{equivalent31}) and (\ref{equivalent32}), by theorem
\ref{wuweiming4} we get main result of this section as following.
\vskip 10pt\noindent
\begin{theorem}\label{wuweiming5}
Assume that  $\delta>0$ or $\delta=0,\gamma-r>\frac{\sigma^{2}}{2}$
and  $0\leq k\leq h(\frac{q}{a})$. Then the value of stock loan with
automatic termination clause, cap and margin is given by (
\ref{solution11}) and (\ref{solution12}). Moreover, if $L> b$ then
the stopping time $ \tau_b\wedge \tau_L$ is the optimal exercise
time. If $L\leq b $ then $ \tau_L$ is the optimal exercise time.
\end{theorem}
\vskip 10pt\noindent
\begin{Remark}
The pricing model (\ref{reward}) or  (\ref{wuweiming5}) resembles
that of American barrier options in mathematical form. If the
pricing model(\ref{reward}) or  (\ref{wuweiming5})has no negative
interest rate, cap and margin constraints, it will  become one of
American barrier options. So the approaches to deal with the pricing
model (\ref{wuweiming5}) and usual American barrier options are very
different because of these constraints. A mathematically oriented
discussion of the barrier option pricing problem is contained in
Rich \cite{Rich}(1994). In general, there are following several
approaches to barrier option pricing: (a) the probabilistic method,
see Kunitomo and Ikeda\cite{KI} (1992), and Mijatovi
\cite{MIJ}(2010); (b) the Laplace Transform technique, see Pelsser
\cite{PE}(2000), Fusai \cite{FU}(2001); (c) the Black-Scholes PDE,
which can be solved using separation of variables, see Hui et
al.\cite{HUI} (2000),Zvan et al.\cite{ZVF}(2000), and
Boyarchenko\cite{BL}(2002) or finite difference schemes and
interpolation, see Boyle and Tian (1998), Sanfelici\cite{S}(2004),
Fusaia and Recchioni\cite{FR}( 2007) and Avrama et
al.\cite{ACU}(2002).; (d) binomial and trinomial trees see Boyle and
Lau \cite{BLA}(1994), Gao, Huang and Subrahmanyam\cite{GHS}(2000);
(e) Monte Carlo simulations with various enhancements, see Baldi et
al. \cite{BCI}(1998), Kudryavtsev and Levendorski\cite{KS} 2009; (f)
variational inequality approach, see Karatzas and  Wang\cite{KW}(
2000 ).
\end{Remark}
\setcounter{equation}{0}
\section{{\small {\bf Ranges of fair values of the parameters}}}
\vskip 5pt \noindent In this section we only work out a ranges of
fair values of the parameters $(q, c, \gamma,a )$ and find
relationships among $q, c, \gamma$ and $a$ based on Theorem
\ref{wuweiming5} and equality $f(S_0)=S_0-q+c$  for stock loan with
automatic termination clause, cap and margin. Another  one can be
similarly treated. Under $\delta>0$ or
$\delta=0,\gamma-r>\frac{\sigma^{2}}{2}$ and $0\leq k\leq
h(\frac{q}{a})$. We distinguish three cases, i.e., $S_{0}\leq a$,
$S_{0}\geq b$ and $a< S_{0}<b$.
 \vskip 15pt
\noindent {\bf Case of  $S_{0}\leq a$}. \quad  By (\ref{solution11})
and $f(S_{0})=S_{0}-q+c$, it has to satisfy $S_{0}-q+c=kS_0$ and so
$c=kS_0 +q-S_{0}$. Since $S_{0}\leq a$, the stock loan is terminated
at the initial time. In this case, the client just sells the stock
to the bank at the initial. The client is reluctant to lose equity
position, hence there is no transaction between the client and the
bank actually. \vskip 15pt

\noindent {\bf Case of  $S_{0}\geq b\wedge L$}. \ The initial value
is $f(S_{0})=S_{0}-q +c$. In order to have $f(S_{0})=S_{0}-q+c$, by
(\ref{solution11}) or (\ref{solution12}),  it must have $S_{0}\wedge
L-q =S_{0}-q+c$. So $c$ must be zero, which means that  the bank
does not charge a service fee for its service since the stock price
is large. By Theorem \ref{wuweiming5} the terminable stopping time
is $\tau_{a}\wedge\tau_{b}=\tau_{b}=0,S_{0}\geq b$. The bank and the
client do not have enough incentive to do the business.\vskip 15pt

\noindent {\bf Case of  $a< S_{0}<b\wedge L$}.\quad In this case
both the client and the bank have incentives to do the business. The
bank does since there is dividend payment and so does the client
since the initial stock price is neither very high nor too low to
trigger the automatic termination clause. By Theorem
\ref{mainresult} the initial value is $f(S_{0})$. Then the bank can
charge an amount $c=f(S_{0})-S_{0}+q$ for its service from the
client. The fair value of the parameters $\gamma,q$, $c$ and $a$
should be such that
\begin{eqnarray}\label{value_determine}
S_{0}-q+c=\frac{ka}{C(a,b\wedge L)}(a^{\frac{\mu}{\sigma}}(b\wedge L
) ^{\frac{\sqrt{\mu^{2}-2\lambda}}{\sigma}}x^{\lambda_{2}}
-a^{\frac{\mu}{\sigma}}
(b\wedge L)^{\frac{-\sqrt{\mu^{2}-2\lambda}}{\sigma}}x^{\lambda_{1}})+\nonumber\\
\frac{b\wedge L-q}{C(a,b\wedge L)}((b\wedge L)^{\frac{\mu}{\sigma}}
a^{\frac{-\sqrt{\mu^{2}-2\lambda}}{\sigma}}x^{\lambda_{1}} -(b\wedge
L)^{\frac{\mu}{\sigma}}a^{\frac{\sqrt{\mu^{2}-2\lambda}}
{\sigma}}x^{\lambda_{2}})
\end{eqnarray}
and the terminable stopping time is
$\tau_{a}\wedge\tau_{b}\wedge\tau_L$ for $a<S_{0}< b\wedge L$.
\vskip 15pt \noindent We determine the fair values by the following
steps: \vskip 5pt \noindent {\sl Step 1 }. \ Determine  the values
$q,a,\gamma$, $k$, $L$ in contract  by negotiation between  the bank
and the client. \vskip 5pt \noindent {\sl Step 2 }. \ compute $b$ by
(\ref{sovle_b1}).\vskip 5pt \noindent
 {\sl Step 3 }. \ Determine service fee $c$ by
(\ref{value_determine}).
 \setcounter{equation}{0}
\section{{\small {\bf Numerical results}}}
\vskip 5pt \noindent In this section we  first consider a stock loan
contract with an automatic termination clause $a$ ($a\in[0,q]$),
$r=0.05,\gamma=0.07,\sigma=0.15,\delta=0.01,q=100$ and $S_{0}=100$.
We will give six numerical examples to show that how the liquidity,
optimal strategy $ b(a)$, initial value $f_a(x)$ and initial cash
$q-c$ depend on automatic termination clause $a$, respectively.
\begin{example}
We see from graph\ref{qca} below that the liquidity obtained with
automatic termination clause is larger than the circumstance without
the automatic termination clause.  When the initial stock price
$S_{0}=100$ and $a=100$, the client just sell the stock to the bank
by the stock loan contract with automatic termination clause. \vskip
15pt \noindent
\begin{figure}[H]
  \includegraphics[width=0.8\textwidth]{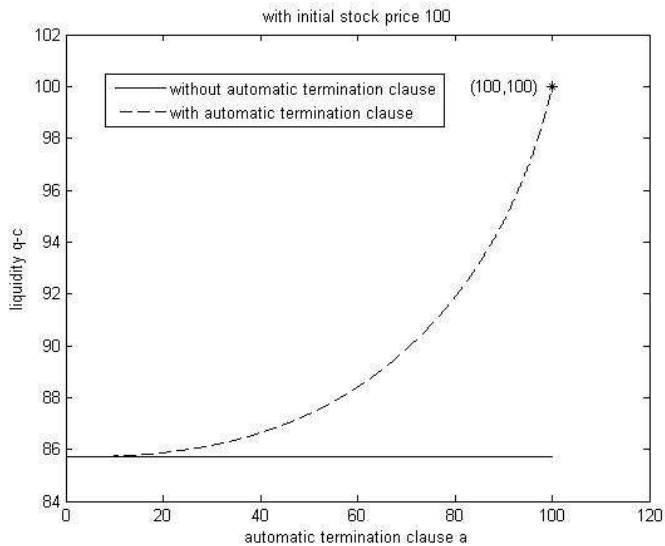}
  \caption{$\gamma=0.07,r=0.05,\sigma=0.15,
  \delta=0.01,q=100,S_{0}=100$}\label{qca}
\end{figure}
\end{example}
\begin{example} We see from graph \ref{bak} below that $b$ is
 an function of $a$. Both the client and the bank will take the deal
when the initial stock price is in between $a$ and $b(a)$. The
client can determine the strategy with automatic termination clause
$a$. The exercise frontier $b(a)$ is decreasing  with respect to
$a$.
\begin{figure}[H]
  \includegraphics[width=0.8\textwidth]{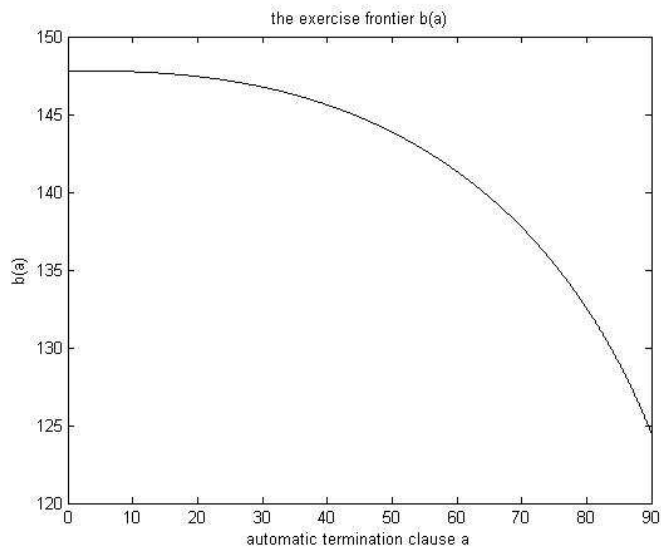}
  \caption{$\gamma=0.07,r=0.05,\sigma=0.15,\delta=0.01,q=100$}\label{bak}
\end{figure}
 \end{example}
\begin{example} The graph \ref{fag} below is a graph of initial value
 $f_{a}(x)$ of the stock loan with
different automatic termination clause $(a=80,60,40,1)$. We see from
the graph that the initial value $f_{a}(x)$ is decreasing w.r.t.
$a$. Since  $c=f(S_{0})-S_{0}+q$, $c$ is also decreasing w.r.t. $a$.
This fact is consistent with the bank can reduce  risk by
introducing an automatic termination clause into the stock loan
contract (see graph \ref{qca}).
\begin{figure}[H]
  \includegraphics[width=0.8\textwidth]{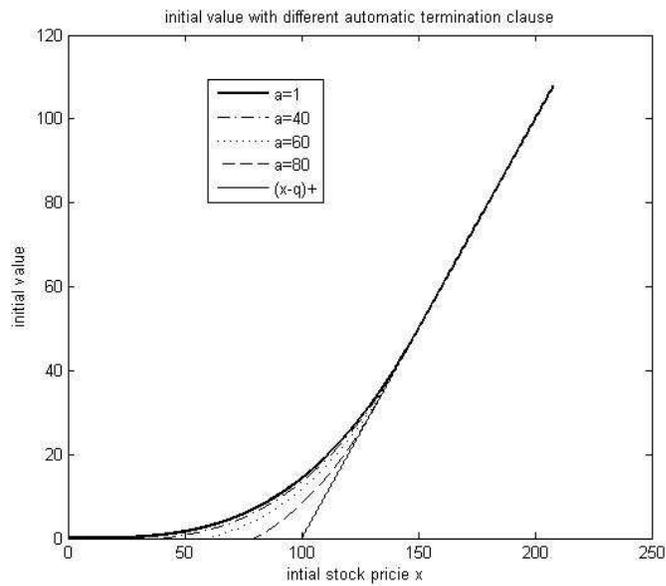}
  \caption{$\gamma=0.07,r=0.05,\sigma=0.15,\delta=0.01,q=100,x=S_{0}$ }\label{fag}
\end{figure}
\end{example}
 \begin{example}
From graph \ref{qc} below  we see  that the initial cash $q-c$ is
increasing  with respect to initial stock price on $[a,b(a)]$. When
the initial stock price is less than $a$, the client just sells the
stock to the bank by the stock loan contract, the bank have no
interest to do the business. In fact  there is no transaction
between the bank and the client.
\begin{figure}[H]
  \includegraphics[width=0.8\textwidth]{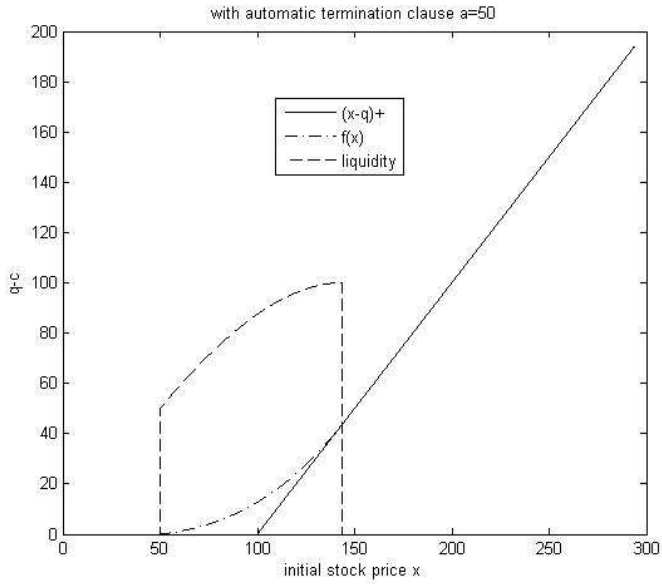}
  \caption{$\gamma=0.07,r=0.05,\sigma=0.15,
  \delta=0.01,q=100,a=50,x=S_{0}$ }\label{qc}
\end{figure}
\end{example}
\vskip 5pt\noindent
 Then we  consider a stock loan contract with automatic
termination clause $a$, cap $L$ and margin $k$.
\begin{example}
The graph \ref{bak} below shows that the  function $b(a, k)$. We see
that for a given contract the client can choose the optimal excise
time.
\begin{figure}[H]
  \includegraphics[width=0.8\textwidth]{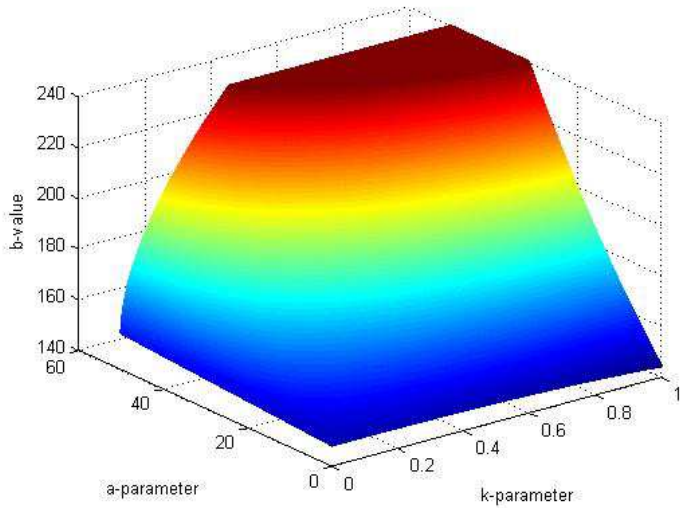}
  \caption{$\gamma=0.07,r=0.05,\sigma=0.15,
  \delta=0.01,q=100,a=10,k=0.5,L=240$}\label{bak}
\end{figure}
\end{example}
\begin{example}
The graphs \ref{fa} and \ref{fa1} below show that the  function
$f_a(x)$. Comparison of the two graphs show the client can get more
flexibility by lower cost.
\begin{figure}[H]
  \includegraphics[width=0.8\textwidth]{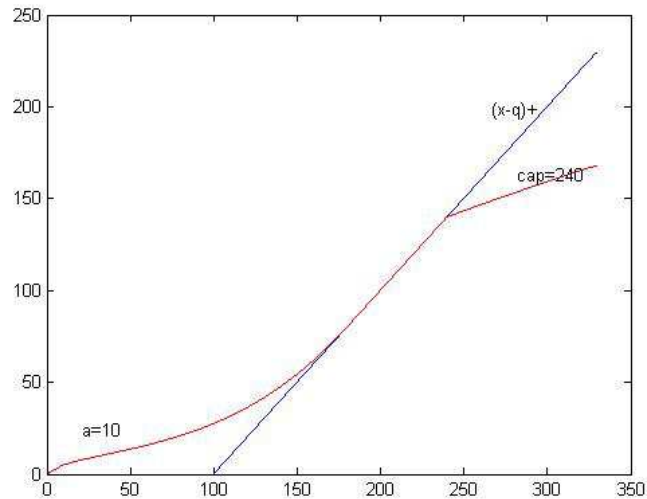}
  \caption{$\gamma=0.07,r=0.05,\sigma=0.15,
  \delta=0.01,q=100,a=10,k=0.5,cap=240$ }\label{fa}
\end{figure}
\begin{figure}[H]
  \includegraphics[width=0.8\textwidth]{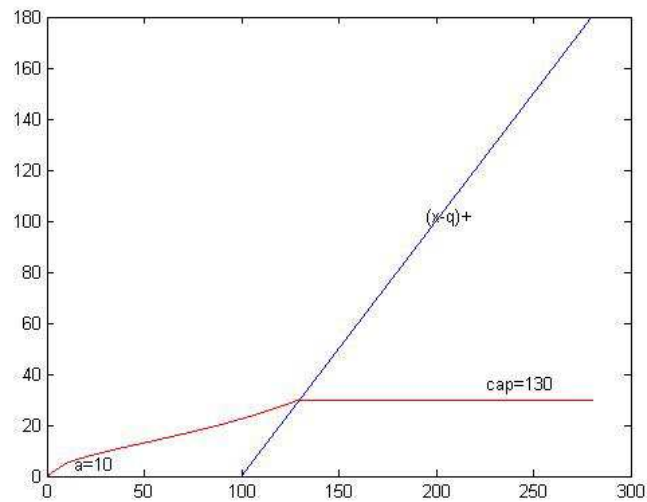}
  \caption{$\gamma=0.07,r=0.05,\sigma=0.15,
  \delta=0.01,q=100,a=10,k=0.5,cap=130$ }\label{fa1}
\end{figure}
\end{example}

 \setcounter{equation}{0}
\section{{\small {\bf Conclusion}}}
\vskip 5pt \noindent In this paper, based on practical transactions
between a bank and a client, we have established a mathematical
model for stock loan with an  automatic termination clause, cap and
margin. The model can be considered a  generalized perpetual
American contingent claim  with possibly negative interest rate. We
have shown that variational inequality method can solve this kind of
stock loans. Using the variational inequality method we have been
able to derive explicitly the value of such a loan, ranges of fair
values of other key parameters, relationships among the key
parameters, and the optimal terminable exercise times. Moreover, we
have checked that the clause $a$, cap $L$ and margin $k$ are
important factors in a stock loan contract by numerical results in
examples 1-6.
 \setcounter{equation}{0}
\section{{\small {\bf Appendix}}}
\vskip 5pt \noindent
\begin{Lemma}\label{AL}
If $\delta>0$ and $\gamma-r+\delta \geq0$, then $\tilde{g}(y)$ is
convex in the region $[\frac{q}{a},\infty)$.
\end{Lemma}
\vskip5pt
 \noindent
\begin{proof} \  It follows from  proof of Proposition
\ref{solution_exist} that there exists $y^{*}$ in the region
$(\frac{q}{a},\infty)$ such that $\tilde{g}(y^*)=0$. Noticing that
\begin{eqnarray*}
\tilde{g}^{''}(y)
&=&(\lambda_{1}-1)(\lambda_{1}+1-\lambda_{2})
(\lambda_{1}-\lambda_{2})y^{\lambda_{1}-\lambda_{2}-1}\\
&&-\frac{q}{a}\lambda_{1}(\lambda_{1}-\lambda_{2})
(\lambda_{1}-\lambda_{2}-1)y^{\lambda_{1}-\lambda_{2}-2}\\
&=&(\lambda_{1}-\lambda_{2})\lambda_{1}(\lambda_{1}-1)y^{\lambda_{1}-2}h(y),\forall
y\geq\frac{q}{a},
\end{eqnarray*}
where $h(y)$ is
\begin{eqnarray}\label{Adx1}
h(y)=\frac{\lambda_{1}+1-\lambda_{2}}{\lambda_{1}}y^{1-\lambda_{2}}
-\frac{q}{a}\frac{\lambda_{1}-1-\lambda_{2}}{\lambda_{1}-1}y^{-\lambda_{2}}\geq
0,\forall y\geq\frac{q}{a},
\end{eqnarray}
it suffices to show that $\tilde{g}^{''}(y)\geq0,y\geq\frac{q}{a}$
for the uniqueness of $y^*$. For this we only need to prove
$h(y)\geq 0, \ \forall y\geq \frac{q}{a}$.  Since
\begin{eqnarray*}
h^{'}(y)=\frac{\lambda_{1}+1-\lambda_{2}}
{\lambda_{1}}(1-\lambda_{2})y^{-\lambda_{2}}
+\frac{q}{a}\frac{\lambda_{1}-1-\lambda_{2}}
{\lambda_{1}-1}\lambda_{2}y^{-\lambda_{2}-1},
\end{eqnarray*}
we  prove (\ref{Adx1})  in  following three cases. \vskip10pt
 \noindent
{\bf Case of $\delta>0,\gamma>r$}. \  In this case we have
$\lambda_{1}>1>\lambda_{2}>0$.
 If $\lambda_{1}-\lambda_{2}\geq 1$, then
$h^{'}(y)\geq0,\ y\geq\frac{q}{a}$. So
\begin{eqnarray*}
h(y)\geq h(\frac{q}{a})>0,y\geq\frac{q}{a}.
\end{eqnarray*}
If $\lambda_{1}-\lambda_{2}<1$, then
\begin{eqnarray*}
h(y)>\frac{\lambda_{1}+1-\lambda_{2}}{\lambda_{1}}y^{1-\lambda_{2}}
\geq\frac{\lambda_{1}+1-\lambda_{2}}{\lambda_{1}}
(\frac{q}{a})^{1-\lambda_{2}}>1,\ y\geq\frac{q}{a}.
\end{eqnarray*}
Therefore (\ref{Adx1}) implies the convexity of $\tilde{g}(y)$.
\vskip10pt
 \noindent
{\bf Case of $\delta>0,\gamma=r$}.\  \   In this case we have
$\lambda_{1}>1>\lambda_{2}=0$ and
\begin{eqnarray*}
h(y)=\frac{\lambda_{1}+1-\lambda_{2}}{\lambda_{1}}y^{1-\lambda_{2}}
-\frac{q}{a}\frac{\lambda_{1}-1-\lambda_{2}}{\lambda_{1}-1}\geq
h(\frac{q}{a})>0,\ y\geq \frac{q}{a}.
\end{eqnarray*}
Obviously, the convexity of $\tilde{g}(y)$ holds. \vskip10pt
 \noindent
{\bf Case of} $\delta>0,\gamma<r$ and $\gamma-r+\delta\geq0$.
 \ \   In this case  we have  $\lambda_{1}>1>0>\lambda_{2}$ and
\begin{eqnarray}  \label{Adx2}
h^{'}(y)&=&\frac{\lambda_{1}+1-\lambda_{2}}
{\lambda_{1}}(1-\lambda_{2})y^{-\lambda_{2}}
+\frac{q}{a}\frac{\lambda_{1}-1-\lambda_{2}}
{\lambda_{1}-1}\lambda_{2}y^{-\lambda_{2}-1}\nonumber\\
&>&\frac{q}{a}y^{-\lambda_{2}-1}(1-\lambda_{2})
(\frac{\lambda_{1}+1-\lambda_{2}}{\lambda_{1}}
-\frac{\lambda_{1}-1-\lambda_{2}}{\lambda_{1}-1}),\
y\geq\frac{q}{a},
\end{eqnarray}
where the last inequality follows from
$\lambda_{1}>1>0>\lambda_{2}>-(1-\lambda_{2})$.

Since $\gamma-r+\delta\geq0$, by (\ref{solve_lumda}) we have
\begin{eqnarray*}
\lambda_{1}+\lambda_{2}=2\frac{\gamma-r+\delta}{\sigma^{2}}+1\geq1.
\end{eqnarray*}
Because
\begin{eqnarray*}
\frac{\lambda_{1}+1-\lambda_{2}}
{\lambda_{1}}-\frac{\lambda_{1}-1-\lambda_{2}}{\lambda_{1}-1}\geq0,
\end{eqnarray*} by (\ref{Adx2}),
 $h^{'}(y)\geq0,\ y>\frac{q}{a}$ and
\begin{eqnarray*}
h(y)\geq h(\frac{q}{a})=(\frac{q}{a})^{-\lambda_{2}}
(\frac{\lambda_{1}+1-\lambda_{2}}{\lambda_{1}}-
\frac{\lambda_{1}-1-\lambda_{2}}{\lambda_{1}-1})\geq0.
\end{eqnarray*}
The convexity holds. Thus we complete the proof.
\end{proof}
 \vskip 15pt\noindent  {\bf Acknowledgements.}    We
are  very grateful to Professor Jianming Xia for his conversation
with us  and providing  original paper of \cite{stock} for us. We
also thank Professor Yongqing Xu for being informed us their work
\cite{Gy}. Special thanks also go to the participants of the seminar
stochastic analysis and finance at Tsinghua University for their
feedbacks and useful conversations. This work is supported by
Projects 10771114 and 11071136 of NSFC, Project 20060003001 of
SRFDP, and SRF for ROCS, SEM, and the Korea Foundation for Advanced
Studies. We would like to thank the institutions for the generous
financial support.
 \vskip 15pt \noindent
\setcounter{equation}{0}

\end{document}